\documentclass{article}
\usepackage[english]{babel}
\usepackage{graphicx,multicol}
\usepackage{epic,eepic,epsfig}
\usepackage{amssymb}
\usepackage{amsmath,amsthm}
\usepackage{color}
\usepackage{tikz}
\usepackage{epsfig,url}
\usepackage{multirow}

\usetikzlibrary{shapes,snakes}

\usepackage{eso-pic} 

\newcommand{\emphX}[1]{\textbf{#1}}

\newcommand{\induce}[2]{\mbox{$ #1 \langle #2 \rangle$}}

\newcommand{\be}{\begin{enumerate}}
\newcommand{\ee}{\end{enumerate}}
\newcommand{\bd}{\begin{description}}
\newcommand{\ed}{\end{description}}

\newcommand{\beq}{\begin{equation}}
\newcommand{\eeq}{\end{equation}}

\newcommand{\2}{\vspace{2mm}}

\renewenvironment{proof}[1][]{\par \noindent {\bf Proof#1}.\ }{\hfill$\Box$
\par \vspace{11pt}}

\newtheorem{theorem}{Theorem}
\newtheorem{lemma}[theorem]{Lemma}
\newtheorem{proposition}[theorem]{Proposition}
\newtheorem{corollary}[theorem]{Corollary}
\newtheorem{claim}{Claim}[theorem]

\theoremstyle{definition}

\newtheorem{question}[theorem]{Question}
\newtheorem{problem}[theorem]{Problem}

\newcommand{\Gcol}{$(\delta^+\geq1, \delta^-\geq 1)$-bipartite-colouring}
\newcommand{\Gpart}{$(\delta^+\geq1, \delta^-\geq 1)$-bipartite-partition}

\begin{document}
\bibliographystyle{plain}

\title{Bipartite spanning sub(di)graphs induced by $2$-partitions}
\author{J. Bang-Jensen$^1$, S. Bessy$^2$, F. Havet$^3$ and A. Yeo$^{1,4}$}
\date{\today}
\maketitle

\begin{center}
{\small $^1$ Department of Mathematics and Computer Science, University of Southern 
Denmark, Odense DK-5230, Denmark, Email: \verb!{jbj,yeo}@imada.sdu.dk!. Financial support: Danish research council, grant number 1323-00178B and Labex UCN@Sophia\\
$^2$ LIRMM, Universit\'e de Montpellier, Montpellier, France.
Email: \verb!stephane.bessy@lirmm.fr!. \\Financial support: OSMO project, 
 Occitanie regional council.\\
$^3$ Universit\'e C\^ote d'Azur, CNRS, I3S  and INRIA, Sophia Antipolis, France.
Email: \verb!frederic.havet@cnrs.fr!. \\
Financial support: ANR-13-BS02-0007 STINT.\\
$^4$ Department of Mathematics,  University of Johannesburg,  Auckland Park, 2006 South Africa. }
\end{center}

\begin{abstract}
 For a given $2$-partition $(V_1,V_2)$ of the vertices of a (di)graph
 $G$, we study properties of the spanning bipartite subdigraph
 $B_G(V_1,V_2)$ of $G$ induced by those arcs/edges that have one end
 in each $V_i$. We determine, for all pairs of non-negative integers
 $k_1,k_2$, the complexity of deciding whether $G$ has a 2-partition
 $(V_1,V_2)$ such that each vertex in $V_i$ has at least $k_i$
 (out-)neighbours in $V_{3-i}$. We prove that it is ${\cal
   NP}$-complete to decide whether a digraph $D$ has a 2-partition
 $(V_1,V_2)$ such that each vertex in $V_1$ has an out-neighbour in
 $V_2$ and each vertex in $V_2$ has an in-neighbour in $V_1$. The
 problem becomes polynomially solvable if we require $D$ to be
 strongly connected. We give a characterisation, based on the
 so-called strong component digraph of a non-strong digraph of the
 structure of ${\cal NP}$-complete instances in terms of their strong
 component digraph. When we want higher in-degree or out-degree
 to/from the other set the problem becomes ${\cal NP}$-complete even
 for strong digraphs. A further result is that it is ${\cal
   NP}$-complete to decide whether a given digraph $D$ has a
 $2$-partition $(V_1,V_2)$ such that $B_D(V_1,V_2)$ is strongly
 connected. This holds even if we require the input to be a highly
 connected eulerian digraph.\\

\noindent{}{\bf Keywords:} $2$-partition, minimum out-degree, spanning
bipartite subdigraph, eulerian, strong spanning subdigraph.
\end{abstract}

\section{Introduction}

A {\bf $2$-partition} of a graph or digraph $G$ is a vertex partition
$(V_1,V_2)$ of its vertex set $V(G)$.  If $(V_1, V_2)$ is a
$2$-partition of a graph (resp. digraph) $G$, the {\bf bipartite
  graph} (resp. {\bf digraph}) {\bf induced by} $(V_1,V_2)$, denoted
by $B_G(V_1,V_2)$, is the spanning bipartite graph (resp. digraph)
induced by the edges (resp. arcs) having one end in each set of the
partition.

The following result is well-known.
\begin{proposition}\label{prop:easy}
 Every undirected graph $G$ admits
a $2$-partition $(V_1,V_2)$ such that $d_{B_G(V_1,V_2)}(v)\geq d_G(v)/2$
for every vertex of $G$. 
\end{proposition}
This proposition implies that every graph with no isolated vertex has
a $2$-partition $(V_1,V_2)$ such that $\delta(B_G(V_1,V_2))\geq 1$.
Consequently, one can decide in polynomial time whether a graph
has a partition such that $d_{B_G(V_1,V_2)}(v)\geq 1$ for all $v$
: if the graph has an isolated vertex, the answer is `No', otherwise
it is `Yes'.  We first study the existence of $2$-partition with some
higher degree constraints on the vertices in the bipartite graph
induced by it.  More precisely, we are interested in the following
decision problem for some fixed positive integers $k_1$ and $k_2$.

\begin{problem}[\sc $(\delta\geq k_1,\delta\geq k_2)$-bipartite-partition]
\label{prob-undir}~\\
\underline{Input}: A graph $G$.\\
\underline{Question}: Does $G$ admit a $2$-partition $(V_1,V_2)$ such that 
$d_{B_G(V_1,V_2)}(v_i)\geq k_i$ for all $v_i\in V_i$, $i\in \{1,2\}$?\\
\end{problem}

As noted above, {\sc $(\delta\geq 1,\delta\geq
  1)$-bipartite-partition} is polynomial-time solvable.  We prove in
Section~\ref{sec:undirected}, that {\sc $(\delta\geq 1,\delta\geq
  2)$-bipartite-partition} is also polynomial-time solvable, and that
{\sc $(\delta\geq k_1,\delta\geq k_2)$-bipartite-partition} is ${\cal
  NP}$-complete when $k_1+k_2\geq 4$.

\medskip

We then consider directed analogues to Problem~\ref{prob-undir}. Many
others $2$-partition problems have already been studied.  The papers
\cite{bangTCS640,bangTCS636} determined the complexity of a large
number of $2$-partition problems where we seek a $2$-partition
$(V_1,V_2)$ with specified properties for the digraphs
$\induce{D}{V_i}$ induced by this partition.  In \cite{bangTCS636} the
authors asked whether there exists a polynomial-time algorithm to
decide whether a given digraph has a $2$-partition $(V_1,V_2)$ with
$\Delta^+(\induce{D}{V_i})< \Delta^+(D)$ for $i=1,2$. This was
answered affirmatively in \cite{bangman17} where also the complexity
of deciding whether a digraph $D$ has a 2-partition $(V_1,V_2)$ so
that $\Delta^+(\induce{D}{V_i})\leq k_i$ was determined for all
non-negative integers $k_1,k_2$.

Thomassen \cite{thomassenEJC6} constructed an infinite class of
strongly connected digraphs ${\cal T}=T_1,T_2,\ldots{},T_k,\ldots{}$
with the property that for each $k$, $T_k$ is $k$-out-regular and has
no even directed cycle. As remarked by Alon in \cite{alonCPC15} this
implies that we cannot expect any directed analogues of
Proposition~\ref{prop:easy}.
\begin{proposition} \label{monochromatic}
 For every $k\geq 1$, for every $2$-partition $(V_1,V_2)$ of $T_k$,
 some vertex $v$ has all its $k$ out-neighbours in the same part as
 itself, so $d^+_{B_D(V_1,V_2)}(v)=0$.
\end{proposition}

The first directed analogue to Problem~\ref{prob-undir} that we study
is the following.
\begin{problem}[\sc $(\delta^+\geq k_1,\delta^+\geq k_2)$-bipartite-partition]
\label{outdegk1k2bippart}~\\
\underline{Input}: A digraph $D$.\\
\underline{Question}: Does $D$ admit a $2$-partition $(V_1,V_2)$ such that 
$d^+_{B_D(V_1,V_2)}(v_i)\geq k_i$ for all $v_i\in V_i$, $i\in \{1,2\}$?
\end{problem}

Observe that, when $k_1=0$ (or $k_2=0$), the problem is pointless
since the partition $(V(D), \emptyset)$ is the desired partition.  We
start in Section~\ref{bipoutdegsec} by using the result of
\cite{bangman17} mentioned to prove that Problem
\ref{outdegk1k2bippart} is polynomial-time solvable when $k_1=k_2=1$
and ${\cal NP}$-complete whenever $k_1+k_2\geq 3$. Then we study the
following problem.

\begin{problem}[\sc $(\delta^+\geq k_1,\delta^-\geq k_2)$-bipartite-partition]~\\
\underline{Input}: A digraph $D$.\\ \underline{Question}: Does $D$
admit a $2$-partition $(V_1,V_2)$ such that every vertex in $V_1$ has
at least $k_1$ out-neighbours in $V_2$ and every vertex in $V_2$ has
at least $k_2$ in-neighbour in $V_1$?
\end{problem}

We show in Section~\ref{sec:inoutbip} that {\sc $(\delta^+\geq
  1,\delta^-\geq 1)$-bipartite-partition} is ${\cal NP}$-complete but
becomes polynomial-time solvable when the input must be a strong
digraph.  We also characterise the ${\cal NP}$-complete instances in
terms of their strong component digraph.  Next, in
Section~\ref{sec:5}, we show that for any pair of positive integers
$(k_1,k_2)$ such that $k_1+k_2\geq 3$, {\sc $(\delta^+\geq
  k_1,\delta^-\geq k_2)$-bipartite-partition} is ${\cal NP}$-complete even when
restricted to strong digraphs.

\medskip

It is simple matter to show that a connected graph $G$ has a
$2$-partition inducing a connected bipartite graph.  Indeed, just
consider a spanning tree and its bipartition.  One can even show that
Theorem~\ref{SbipconG} below holds\footnote{Stephan Thomass\'e private
  communication, Lyon 2015.} (just consider a $2$-partition maximizing
the number of edges between the two sets). Recall that $\lambda(G)$ is
the {\bf edge-connectivity} of $G$, that is, the minimum number of
edges whose removal from $G$ results in a disconnected graph.

\begin{theorem}
\label{SbipconG}
Every graph $G$ has a $2$-partition $(V_1,V_2)$ such that
$\lambda(B_G(V_1,V_2))\geq \lfloor{}\lambda(G)/2\rfloor$.
\end{theorem}

We thus study the directed analogues, called {\bf strong
  $2$-partitions}, which are $2$-partitions $(V_1,V_2)$ of a digraph
$D$ such that $B_D(V_1,V_2)$ is strong.  It is a well-known phenomenon
that results on edge-connectivity for undirected graphs often have a
counterpart for eulerian digraphs. Unfortunately, we show that it is
not the case for Theorem~\ref{SbipconG} : for every $r>0$, there
exists an $r$-strong eulerian digraph $D$ which has no strong
$2$-partition (Theorem~\ref{thm1}).  We then show that the following
problem is ${\cal NP}$-compete even when restricted to $r$-strong
digraphs (for some $r>0$).

\begin{problem}[\sc Strong $2$-partition]~\\
  \underline{Input}: A digraph $D$.\\
  \underline{Question}: Does $D$
admit a $2$-partition $(V_1,V_2)$ such that $B_D(V_1,V_2)$ is strong?
\end{problem}

We conclude the paper with a section presenting some remarks and open problems.

\section{Notation}

Notation follows \cite{bang2009}.  We use the shorthand notation $[k]$
for the set $\{1,2,\ldots{},k\}$.  Let $D=(V,A)$ be a digraph with
vertex set $V$ and arc set $A$.

Given an arc $uv\in A$, we say that $u$ \emphX{dominates} $v$ and $v$
is \emphX{dominated} by $u$. If $uv$ or $vu$ (or both) are arcs of
$D$, then $u$ and $v$ are {\bf adjacent}.  If neither $uv$ or $vu$
exist in $D$, then $u$ and $v$ are {\bf non-adjacent}.  The {\bf
  underlying graph} of a digraph $D$, denoted by $UG(D)$, is obtained
from $D$ by suppressing the orientation of each arc and deleting
multiple copies of the same edge (coming from directed 2-cycles).  A
digraph $D$ is {\bf connected} if $UG(D)$ is a connected graph, and
the {\bf connected components} of $D$ are those of $UG(D)$.

The \emphX{subdigraph induced} by a set of vertices $X$ in a digraph
$D$, denoted by \induce{D}{X}, is the digraph with vertex set $X$ and
which contains those arcs from $D$ that have both end-vertices in
$X$. When $X$ is a subset of the vertices of $D$, we denote by $D-X$
the subdigraph $\induce{D}{V\setminus X}$. If $D'$ is a subdigraph of
$D$, for convenience we abbreviate $D-V(D')$ to $D-D'$. The
\emphX{contracted digraph} $D/X$ is obtained from $D-X$ by adding a
`new' vertex $x$ not in $V$ and by adding for every $u \in D-X$ the
arc $ux$ (resp. $xu$) if $u$ has an out-neighbour (resp. in-neighbour)
in $X$ (in $D$).

The \emphX{in-degree} (resp. \emphX{out-degree}) of $v$, denoted by
$d^-_D(v)$ (resp. $d^+_D(v)$), is the number of arcs from $V\setminus
\{v\}$ to $v$ (resp. $v$ to $V\setminus \{v\}$).  A {\bf sink} is a
vertex with out-degree $0$ and a {\bf source} is a vertex with
in-degree $0$.  The {\bf degree} of $v$, denoted by $d_D(v)$, is given
by $d_D(v)=d^+_D(v)+d^-_D(v)$.  Finally the \emphX{minimum
  out-degree}, respectively \emphX{minimum in-degree} and
\emphX{minimum degree} is denoted by $\delta^+(D)$, respectively
$\delta^-(D)$ and $\delta(D)$.  The \emphX{minimum semi-degree} of
$D$, denoted by $\delta^0(D)$, is defined as
$\delta^0(D)=\min\{\delta^+(D),\delta^-(D)\}$.  A vertex is {\bf
  isolated} if it has degree $0$.

\noindent{}A digraph is {\bf $k$-out-regular} if all its vertices have
out-degree $k$.

A {\bf $(u,v)$-path} is a directed path from $u$ to $v$.  A digraph is
{\bf strongly connected} (or {\bf strong}) if it contains a
$(u,v)$-path for every ordered pair of distinct vertices $u,v$.  A
digraph $D$ is {\bf $k$-strong} if for every set $S$ of less than $k$
vertices the digraph $D-S$ is strong.  A {\bf strong component} of a
digraph $D$ is a maximal subdigraph of $D$ which is strong. A strong
component is {\bf trivial}, if it has order $1$. An {\bf initial}
(resp. {\bf terminal}) strong component of $D$ is a strong component
$X$ with no arcs entering (resp. leaving) $X$ in $D$.

Let $D$ be a strongly connected digraph. If $S$ is a strong subdigraph
of $D$, then an {\bf $S$-handle} $H$ of $D$ is a directed walk
$(s,v_1,\ldots,v_\ell,t)$ such that:
\begin{itemize}
\item the $v_i$ are distinct and in $V(D-S)$, and
\item $s,t\in V(S)$ (with possibly $s=t$ and in this case $H$ is a
  directed cycle, otherwise it is a directed path).
\end{itemize}

The {\bf length} of a handle is the number of its arcs, here $\ell+1$.
A handle of length one is said to be {\bf trivial}. 


An {\bf out-tree} rooted at the vertex $s$, also called an {\bf
  $s$-out-tree}, is a connected digraph $T^+_s$ such that
$d^-_{T^+_s}(s) =0$ and $d^-_{T^+_s}(v)=1$ for every vertex $v$
different from $s$.  Equivalently, for every $v\in V(T^+_s)$ there is a unique $(s,v)$-path in $T^+_s$.  The directional
dual notion is the one of an {\bf $s$-in-tree}, that is, a connected
digraph $T^-_s$ such that $d^+_{T^-_s}(s) =0$ and $d^+_{T^-_s}(v)=1$
for every vertex $v$ different from $s$.

An {\bf ${s}$-out-branching} (resp. {\bf $s$-in-branching}) is a
spanning $s$-out-tree (resp. $s$-in-tree).  We use the notation
$B^+_s$ (resp. $B^-_s$) to denote an $s$-out-branching (resp. an
$s$-in-branching). \\

In our ${\cal NP}$-completeness proofs we use reductions from the
well-known 3-SAT problem,  and two variants {\sc Not-All-Equal-3-SAT}, and  {\sc
  Monotone Not-all-equal-3-SAT}. In the later the
Boolean formula ${\cal F}$ consists of clauses all of
whose literals are non-negated variables. 
In 3-SAT, we want to decide whether there is a truth assignment that {\bf satisfies ${\cal F}$} that is such that every clause has a true literal.
In {\sc Not-All-Equal-3-SAT} and  {\sc
  Monotone Not-all-equal-3-SAT}, we want to decide whether there is a {\bf NAE truth assignment} , that is a truth assigment such that every clause has a true literal and a false literal.
  Those two problems are ${\cal NP}$-complete
\cite{shaeferSTOC10}.

Let ${\cal P}_1,{\cal P}_2$ be properties of vertices in a
digraph (e.g. out-degree at least $1$). Then a {\bf $({\cal P}_1,{\cal
    P}_2)$-bipartite-partition} of a graph $D$ is a $2$-partition
$(V_1,V_2)$ such that the vertices of $V_i$ have property ${\cal P}_i$
in $B_D(V_1,V_2)$. For example, a $(\delta^+\geq 1,\delta^-\geq
1)$-bipartite-partition is a $2$-partition $(V_1,V_2)$ so that in
$B_D(V_1,V_2)$ the vertices of $V_1$ have out-degree at least $1$ and
the vertices of $V_2$ have in-degree at least $1$.  We also use the
same definition for (undirected) graphs.
The {\bf $2$-colouring associated to} a $2$-partition $(V_1,V_2)$ is the $2$-colouring $c$ defined by $c(x)=i$ if $x\in V_i$. 
A {\bf $({\cal P}_1,{\cal P}_2)$-colouring} is a $2$-colouring associated to a $({\cal P}_1,{\cal
    P}_2)$-bipartite-partition.

\section{$(\delta\geq k_1,\delta\geq k_2)$-bipartite-partitions}\label{sec:undirected}

In this section we give a complete characterisation of the complexity
of the $(\delta\geq k_1,\delta\geq k_2)$-bipartite-partition problem
for undirected graphs.  We first list an easy consequence of
Proposition \ref{prop:easy}.

\begin{corollary}
Every graph $G$ with $\delta{}(G)\geq 1$ has a $(\delta\geq
1,\delta\geq 1)$-bipartite-partition and every graph $G$ with at least
one edge has a $(\delta\geq 0,\delta\geq 1)$-bipartite-partition.
\end{corollary}

This statement can be generalized to $(\delta\geq
1,\delta\geq k)$-bipartite-partitions.

\begin{theorem}\label{thm:iff2}
If $G$ is a graph with $\delta(G) \geq k$, then $G$ has a $(\delta\geq
1,\delta\geq k)$-bipartite-partition, which can be found in polynomial
time.
\end{theorem}
\begin{proof}
Let $V_2$ be any maximal stable set in $G$. That is, for every $x
\not\in V_2$ the set $V_2 \cup \{x\}$ is not stable.  This
implies that every vertex not in $V_2$ has an edge to a vertex in
$V_2$ and as $V_2$ is stable and $\delta(G) \geq k$ every vertex
in $V_2$ has at least $k$ neighbours not in $V_2$. Therefore $(V(G)
\setminus V_2,V_2)$ is the desired partition. As a maximal stable set may be computed greedily,
the partition $(V(G) \setminus V_2,V_2)$ can be found
  in polynomial time.
\end{proof}

\subsection{Solving {\sc $(\delta\geq 1,\delta\geq  2)$-bipartite-partition} in polynomial time} 

\begin{theorem}\label{thm:iff}
Let $G$ be a graph with $\delta(G)=1$. Let $S_1$ be the set of
vertices of degree $1$ in $G$.  Then $G$ has a $(\delta\geq
1,\delta\geq 2)$-bipartite-partition if and only if $S_1$ is a stable
set and every vertex in $N(S_1)$ has either two neighbours in $S_1$ or
at least one neighbour in $V(G)\setminus (S_1\cup N(S_1))$.
\end{theorem}
\begin{proof}
Suppose that $G$ has a $(\delta\geq 1,\delta\geq
2)$-bipartite-partition $(V_1,V_2)$. Necessarily, $S_1\subseteq
V_1$. Moreover, for each $v\in S_1$, its unique neighbour is in
$V_2$. Hence $S_1$ is a stable set and $N(S_1)\subseteq V_2$.  Now
every vertex in $N(S_1)$ has at least two neighbours in $V_1$. Hence
either two neighbours are in $S_1$ or at least one neighbour is in
$V_1\setminus S_1$ which is a subset of $V(G)\setminus (S_1\cup
N(S_1))$.

\medskip

Reciprocally, assume that $S_1$ is a stable set and that every vertex
in $N(S_1)$ has either two neighbours in $S_1$ or one neighbour in
$V(G)\setminus (S_1\cup N(S_1))$.  For every $i>1$, let $S_i$ be the
set of vertices not in $\bigcup_{j=1}^{i-1}S_j$ which are adjacent to
a vertex in $S_{i-1}$. Note that $S_2=N(S_1)$.  Moreover, for every
vertex $v$ in $S_i$, its {\bf predecessors} (resp. {\bf peers}, {\bf
  successors}) are its neighbours in $S_{i-1}$ (resp. $S_{i}$,
$S_{i+1}$).  By definition of the $S_i$, every vertex in
$V(G)\setminus S_1$ has at least one predecessor.

We initially colour the vertices as follows: if $v\in S_i$ and $i$ is
odd, then $v$ is coloured $1$, otherwise it is coloured $2$.  Observe
that a vertex has a colour different from that of its predecessors and
successors.  Now as long as there is a vertex $w$ coloured $2$ with
exactly one neighbour coloured $1$, we recolour $w$ with $1$.  Let $w$
be a recoloured vertex. As it is originally coloured $2$, it must be
in $S_i$ with $i$ even. Now $w$ has exactly one predecessor and no
successor.  In particular, it is not in $S_2$, by our assumption on
$N(S_1)=S_2$. Furthermore, it has degree at least $2$ (since vertices
of $S_1$ are coloured $1$), so its has at least one peer which must be
coloured $2$, and will never be recoloured because it now has at least
two neighbours (a peer and a predecessor) coloured $1$.

Let $V_1$ (resp. $V_2$) be the set of vertices coloured $1$
(resp. $2$). We claim that $(V_1, V_2)$ is a $(\delta\geq 1,\delta\geq
2)$-bipartite-partition.\\ Consider a vertex $v_1$ in $V_1$. Assume
$v_1$ is originally coloured $1$. Either it is in $S_1$ and its
neighbour is in $S_2$ and thus in $V_2$ because no vertex of $S_2$ is
recoloured, or it has a predecessor which must be in $V_2$ because
only vertices with no successors are recoloured.  If $v_1$ has been
recoloured, then as observed above it has a peer originally coloured
$2$ that is not recoloured.\\ Consider now a vertex $v_2\in V_2$. It
was originally coloured $2$ and has not been recoloured. Hence $v_2$
has at least two neighbours coloured $1$.
\end{proof}

We note that any graph with an isolated vertex does not contain a
$(\delta\geq 1,\delta\geq k)$-bipartite-partition for any $k \geq 1$.
In Theorem~\ref{thm:iff} we consider graphs with $\delta(G)=1$. The
following easy result handles the cases when $\delta(G) \geq 2$ as a
special case (when $k=2$).

\begin{corollary}
One can decide in polynomial time whether a given graph has a
$(\delta\geq 1,\delta\geq 2)$-bipartite-partition.  Moreover if such a
partition exists, it can be found in polynomial time .
\end{corollary}
\begin{proof}
Let $G$ be a graph and let $S_1$ be the set of vertices with degree
$1$ in $G$.  If $S_1$ has an isolated vertex then no such partition
exists, and if $S_1=\emptyset$ then the result follows from
Theorem~\ref{thm:iff}, so assume that $S_1 \not= \emptyset$ and $G$
does not contain any isolated vertices.  According to
Theorem~\ref{thm:iff}, deciding whether a graph $G$ has a
$(\delta\geq 1,\delta\geq 2)$-bipartite-partition, we need to check
that $S_1$ is a stable set, and that every vertex in $N(S_1)$ has
either two neighbours in $S_1$ or one neighbour in $V(G)\setminus
(S_1\cup N(S_1))$.  This can easily be done in polynomial time.

Moreover, since the proof of Theorem~\ref{thm:iff} is constructive,
one can find in polynomial time a $(\delta\geq 1,\delta\geq
2)$-bipartite-partition, if one exists.
\end{proof}

\subsection{${\cal NP}$-completeness of {\sc $(\delta\geq k_1,\delta\geq  k_2)$-bipartite-partition} when $k_1+k_2\geq 4$}

\begin{theorem}
Let $k_1, k_2$ be integers such that $2\leq k_1\leq k_2$. It is ${\cal
  NP}$-complete to decide whether a graph $G$ has a $(\delta\geq k_1,\delta\geq  k_2)$-bipartite-partition.
\end{theorem}

\begin{proof}
We reduce  {\sc Not-all-equal-3-SAT} to the problem of deciding whether a graph has a $(\delta\geq k_1,\delta\geq  k_2)$-colouring (which is equivalent to {\sc $(\delta\geq k_1,\delta\geq  k_2)$-bipartite-partition}.

A $2$-colouring is {\bf good for $X$} if every vertex of $X$ coloured $i$ has $k_i$ neighbours coloured $3-i$.   
 In particular, a colouring good for $V(G)$ is a $(\delta\geq k_1,\delta\geq  k_2)$-colouring of $G$.
   
   First we define some gadgets and
then we show how to combine them to produce the desired result.\\

Let $X'$ be the graph whose vertex set is the disjoint union of
the sets $\{v,z,x,\bar{x}\}, X_1,X_2,X_3,X_4$, where
$|X_1|=|X_4|=k_1-1$ and $|X_2|=|X_3|=k_2-1$.  The graph $X'$ has the
following edges (when we write `all edges', we mean all possible edges
between the two sets): All edges between $v$ and $X_1$, all edges
between $X_1$ and $X_2$, all edges between $X_2$ and $\{x,\bar{x}\}$,
all edges between $\{x,\bar{x}\}$ and $X_4$, all edges between $X_4$
and $X_3$, all edges between $X_3$ and $z$ and finally the edge
$vz$. Let $X$ be obtained from $X'$ by adding the edge $x\bar{x}$.  It
is easy to verify that $X$ has a $(\delta\geq k_1,\delta\geq  k_2)$-colouring and in every such
colouring of $X$ the vertices $x$ and $\bar{x}$ must get different
colours and both colourings are possible. Indeed once we fix the colour
of $v$, which must be 1 if $k_1<k_2$ and could be 1 or 2 if $k_1=k_2$,
then every other colour is fixed except for $x$ and $\bar{x}$.
Moreover, this property remains no  matter what edges we add to $\{x, \bar{x}\}$.

\medskip

Let $Y$ be the graph that we obtain from a copy of $X'$ by adding
three new vertices $\ell_1,\ell_2,\ell_3$ and all possible edges
between these and the set $\{x,\bar{x}\}$. As previously it is easy to
verify that $Y$ has a $2$-colouring goo for all vertices
except $\ell_1,\ell_2,\ell_3$ (they do not have enough neighbours in
$Y$ but will get these in the full graph we construct) and for every
such colouring at least one of $\ell_1,\ell_2,\ell_3$ is coloured $i$
for $i\in [2]$. Furthermore, every colouring of
$\{\ell_1,\ell_2,\ell_3\}$ where both colours are used can be extended
to a full $2$-colouring which is good for
$V(Y)\setminus \{\ell_1,\ell_2,\ell_3\}$

\medskip

Now let ${\cal F}$ be an instance of {\sc Not-all-equal-3-SAT} with
variables $v_1,v_2,\ldots{},v_n$ and clauses $C_1,C_2,\ldots{},C_m$.
Form a graph $G=G({\cal F})$ from ${\cal F}$ as follows: make $n$
disjoint copies $X_1,X_2,\ldots{},X_n$ of $X$ and denote the copies of
$x,\bar{x}$ in $X_i$ by $x_i,\bar{x}_i$. Then make $m$ disjoint copies
$Y_1,Y_2,\ldots{},Y_m$ of $Y$, where the $j$th copy will correspond to
the clause $C_j$. Denote the copies of $\ell_1,\ell_2,\ell_3$ in $Y_j$
by $\ell_{j,1},\ell_{j,2},\ell_{j,3}$.  Now, for each $j\in [m]$
identify the vertices $\ell_{j,1},\ell_{j,2},\ell_{j,3}$ with those
vertices from $Z=\{x_1,\bar{x}_1,\ldots{},x_n,\bar{x}_n\}$ that
correspond to the literals of $C_j$. E.g. if $C_j=(v_1\vee
\bar{v}_3\vee v_7)$ then we identify $\ell_{j,1}$ with $x_1$,
$\ell_{j,2}$ with $\bar{x}_3$ and $\ell_{j,3}$ with $x_7$. Note that
each vertex from $Z$ may be identified with many vertices in this
way.\\

We claim that $G$ has a $(\delta\geq k_1,\delta\geq  k_2)$-colouring if and only if there is NAE truth assignment for ${\cal F}$.
 Suppose first that
$c$ is a $(\delta\geq k_1,\delta\geq  k_2)$-colouring of $G$.  Let $\phi$ be the truth assignment which sets $v_i$ true precisely when
$c(x_i)=1$. By the property of the vertices $\ell_1,\ell_2,\ell_3$
(which is inherited in all the subgraphs $Y_1,\ldots{},Y_m$) and the
fact that $c$ is a $(\delta\geq k_1,\delta\geq  k_2)$-colouring implies that for each $j$ the number of vertices
from $\{\ell_{j,1},\ell_{j,2},\ell_{j,3}\}$ that have colour 1 is
either one or two. Moreover by the property of $X$ for every
$i=1,\dots ,n$ the vertices $x_i$ and $\bar{x}_i$ receive different
colours by $c$. So $\phi$ is a NAE truth assignment. \\
Conversely, if $\phi$
is a NAE truth assignment, then we first colour each $x_i$ by 1
and $\bar{x}_i$ by 2 if $\phi(v_i)$ is true and do the opposite
otherwise. If is easy to check from the definition of $X,Y$ that we can
extend this partial 2-colouring to a $(\delta\geq k_1,\delta\geq  k_2)$-colouring of $G$.
\end{proof}

Recall Theorem~\ref{thm:iff2} which states that if $G$ is a graph with
$\delta(G) \geq k$ then $G$ has a $(\delta\geq 1,\delta\geq
k)$-bipartite-partition.  In contrast to this result we prove the
following result.

\begin{theorem}
For all integers $k \geq 3$ it is ${\cal NP}$-complete to decide
whether a graph $G$ has a $(\delta\geq 1,\delta\geq
k)$-bipartite-partition.  In fact the problem remains ${\cal
  NP}$-complete even for graphs $G$ with $\delta(G)=k-1$.
\end{theorem}

\begin{proof} 
Reduction from 3-SAT.
$(\delta\geq 1,\delta\geq k)$-bipartite-partition.

First suppose that $k=3$. Let the gadget $G^*$ contain the vertices
$\{a_1,a_2,x,\bar{x},y_1,y_2,b_1,b_2\}$ and all edges from
$A=\{a_1,a_2\}$ to $X=\{x,\bar{x}\}$ and all edges from $X$ to
$Y=\{y_1,y_2\}$ and all edges from $Y$ to $B=\{b_1,b_2\}$.  
See Figure~\ref{figureGs}.
\begin{figure}[h!]
\begin{center}
\tikzstyle{vertexX}=[circle,draw, top color=gray!5, bottom
  color=gray!30, minimum size=16pt, scale=0.6, inner sep=0.5pt]
\tikzstyle{vertexZ}=[circle,draw, top color=white!5, bottom
  color=white!30, minimum size=4pt, scale=0.6, inner sep=0.5pt]
\tikzstyle{vertexBIG}=[ellipse, draw, scale=0.6, inner sep=3.5pt]
\begin{tikzpicture}[scale=0.5]

 \node (a1) at (1.0,3.0) [vertexX] {$a_1$};
 \node (a2) at (1.0,1.0) [vertexX] {$a_2$};
 \node (x1) at (4.0,3.0) [vertexX] {$x$};
 \node (x2) at (4.0,1.0) [vertexX] {$\bar{x}$};
 \node (y1) at (7.0,3.0) [vertexX] {$y_1$};
 \node (y2) at (7.0,1.0) [vertexX] {$y_2$};
 \node (b1) at (10.0,3.0) [vertexX] {$b_1$};
 \node (b2) at (10.0,1.0) [vertexX] {$b_2$};

    \draw [line width=0.05cm] (a1) -- (x1);
    \draw [line width=0.05cm] (a1) -- (x2);
    \draw [line width=0.05cm] (a2) -- (x1);
    \draw [line width=0.05cm] (a2) -- (x2);

    \draw [line width=0.05cm] (y1) -- (x1);
    \draw [line width=0.05cm] (y1) -- (x2);
    \draw [line width=0.05cm] (y2) -- (x1);
    \draw [line width=0.05cm] (y2) -- (x2);

    \draw [line width=0.05cm] (y1) -- (b1);
    \draw [line width=0.05cm] (y1) -- (b2);
    \draw [line width=0.05cm] (y2) -- (b1);
    \draw [line width=0.05cm] (y2) -- (b2);

\end{tikzpicture}
\end{center}
\caption{The gadget $G^*$.}
\label{figureGs}
\end{figure}
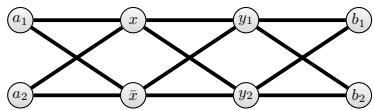

No matter
what edges we latter add to $X$ we note that the vertices in $A \cup
B$ must receive colour $1$ (as they have degree $2$) in any $(\delta\geq 1,\delta\geq 3)$-colouring. Furthermore at
least one vertex in $Y$ must get colour $2$ (as the vertices in $B$
needs a neighbour of colour $2$). Without loss of generality, assume
$y_1$ has colour $2$.  Due to the vertices in $A$ one vertex in $X$
must be coloured $2$ and due to $y_1$ one vertex in $X$ must be coloured
$1$. So the vertices in $X$ must receive different colours in any $(\delta\geq 1,\delta\geq 3)$-colouring. 
Furthermore if we do colour exactly one vertex from $X$ and
one vertex from $Y$ and all vertices in $A \cup B$ with the colour $1$
then we get a $(\delta\geq 1,\delta\geq 3)$-colouring of the gadget $G^*$.  

\smallskip

Let ${\cal F}$ be an instance of 3-SAT with variables $v_1,v_2,\ldots
,v_n$ and clauses $C_1,C_2,\ldots ,C_m$.  Form a graph $G=G({\cal F})$
from ${\cal F}$ as follows: make $n$ disjoint copies
$X_1^*,X_2^*,\ldots{},X_n^*$ of $G^*$ and denote the copies of
$x,\bar{x}$ in $X_i^*$ by $x_i,\bar{x}_i$.  Then add $m$ disjoint
copies of $3$-cycles with vertex sets $y_j,y_j',y_j''$ for $j \in
[m]$.  Now, for each $j\in [m]$ add an edge from $y_j$ to those
vertices from $Z=\{x_1,\bar{x}_1,\ldots{},x_n,\bar{x}_n\}$ that
correspond to the literals of $C_j$.  E.g. if $C_j=(v_1\vee
\bar{v}_3\vee v_7)$, then we add an edge from $y_j$ to $x_1$,
$\bar{x}_3$ and $x_7$.  As in a $(\delta\geq 1,\delta\geq 3)$-colouring $y_j'$ and $y_j''$ must
be given colour $1$ we note that there is a $(\delta\geq 1,\delta\geq 3)$-colouring of $G$ if
each $y_j$ (which has to be coloured $2$) has a neighbour in $Z$ of
colour $1$. It is now not difficult to see that $G$ has a $(\delta\geq 1,\delta\geq 3)$-colouring
if and only if ${\cal F}$ is satisfied (where the vertex $x_i$ is
given colour $1$ if the variable $v_i$ is true and otherwise
$\bar{x}_i$ is given colour $1$).  Furthermore we note that
$\delta(G')=2=k-1$ by construction.

\medskip

We now consider the case when $k \geq 4$. We will reduce from the case
when $k=3$ as follows.  Let $G$ be an instance of the case when $k=3$
and now assume that $k \geq 4$. Let $X_1,X_2,\ldots, X_{k-3}$ be $k-3$
cliques of size $k$ and let $x_i \in X_i$ be arbitrary for $i \in
[k-3]$.  Let $G'$ be the graph obtained from $G$ by adding
$X_1,X_2,\ldots ,X_{k-3}$ and the vertices
$\{y_1,y_2,\ldots,y_{k-3}\}$ to $G$ and all edges $x_iy_i$ and all
edges from $y_i$ to $V(G)$ for all $i \in [k-3]$.  Note that if $(V_1,V_2)$ is a $(\delta\geq 1,\delta\geq k)$-bipartite-partition of $G'$ the vertices in $V(X_i) \setminus \{x_i\}$ must be in $V_1$ (as they have degree $k-1$) for $i \in [k-3]$ and
therefore $x_1,x_2,\ldots,x_{k-3}$ must be in $V_ 2$ and
$y_1,y_2,\ldots,y_{k-3}$ must be in $V_1$ (as $d(x_i)=k$ for $i \in
[k-3]$).  This implies that $G$ admits a
$(\delta\geq 1,\delta\geq 3)$-bipartite-partition if and only if $G'$
admits a $(\delta\geq 1,\delta\geq k)$-bipartite-partition.  This
completes the proof as we note that $\delta(G')=k-1$ (as
$\delta(G)=2$).
\end{proof}


\section{$(\delta^+\geq k_1,\delta^+\geq k_2)$-bipartite-partitions}\label{bipoutdegsec}

We now use results from \cite{bangman17} to settle the complexity of
the $(\delta^+\geq k_1,\delta^+\geq k_2)$-bipartite-partition problem
for digraphs for all pairs of positive integers $k_1,k_2$.

The following result was proved by the authors in
\cite{bangman17}. Note that one can find an even cycle in a digraph
that has such a cycle in polynomial time \cite{robertsonAM150}.
\begin{theorem}[\cite{bangman17}]\label{thm:struc-k1-p1}
A digraph $D$  admits a $(\delta^+\geq 1,\delta^+\geq 1)$-bipartite-partition if
  and only if every non-trivial terminal strong component contains an
  even directed cycle.
The desired $2$-partition can be
constructed in polynomial time when it exists.
\end{theorem}

We now show that for all other positive values of $k_1,k_2\geq 1$ the
{\sc $(\delta^+\geq k_1,\delta^+\geq k_2)$-bipartite-partition} problem is
${\cal NP}$-complete. In fact, this remains true even if the input is
strong and out-regular.

\begin{theorem} \label{npc_proof1}
 Let  $k_1,k_2$ be positive integers such that $k_1+k_2\ge 3$. Then
  {\sc $(\delta^+\geq k_1,\delta^+\geq k_2)$-bipartite-partition} is
  ${\cal NP}$-complete. It remains ${\cal NP}$-complete when the input
  is required to be strongly connected and out-regular.
\end{theorem}

\begin{proof}
In \cite{bangman17} it was shown that deciding the existence of a
2-partition $(V_1,V_2)$ of a digraph $D$ so that
$\Delta^+(\induce{D}{V_1})\leq a_1$ and $\Delta^+(\induce{D}{V_2})\leq
a_2$ is ${\cal NP}$-complete for all $a_1,a_2$ with $\max\{a_1,a_2\}\ge
1$ even when the input is a strong out-regular digraph.  More
precisely, when $a_1=a_2$ the problem is ${\cal NP}$-complete for
strong $p$-out-regular digraphs when $p\geq a_1+2$ and when $a_1<a_2$,
the problem is ${\cal NP}$-complete for strong $p$-out-regular
digraphs when $p\geq a_2+1$.  The first of these results implies that
{\sc $(\delta^+\geq k,\delta^+\geq k)$-bipartite-partition} is
${\cal NP}$-complete for strong $(k+2)$-out-regular digraphs.  The
second result implies that {\sc $(\delta^+\geq k_1,\delta^+\geq
k_2)$-bipartite-partition problem} with $1\leq k_1<k_2$ is ${\cal
  NP}$-complete for strong $(k_2+1)$-out-regular digraphs.
\end{proof}

\section{$(\delta^+\geq 1,\delta^-\geq 1)$-bipartite-partitions}\label{sec:inoutbip}

We now turn to the complexity of deciding whether a given digraph has a
$(\delta^+\geq 1,\delta^-\geq 1)$-bipartite-partition. We first show
that this problem is ${\cal NP}$-complete for acyclic digraphs but
polynomial-time solvable when the input is a strong digraph. Then we
classify, in terms of the strong component digraph, those classes of
non-strong digraphs for which the problem is ${\cal NP}$-complete. For
all the remaining classes it turns out that a partition always exists.

\begin{theorem}\label{delta+-NPC}
It is ${\cal NP}$-complete to decide whether an acyclic digraph has a
$(\delta^+\geq 1,\delta^-\geq 1)$-bipartite-partition.
\end{theorem}
\begin{proof}
 Let ${\cal F}$ be an instance of 3-SAT
 with clauses $C_1,C_2,\ldots{},C_m$ over the variables
 $v_1,v_2,\ldots{},v_n$. Form a digraph $M=M({\cal F})$ as follows:
 \begin{eqnarray*}
 V(M) & = & \bigcup_{i\in [n]}  \{x_i, \bar{x}_i, y_i, z_i\} \cup \{c_j \mid j\in [m]\},\\
  E(M) & = & \bigcup_{i\in [n]} \{y_ix_i, y_i\bar{x}_i, x_iz_i, \bar{x}_iz_i\} \cup \{x_ic_j \mid x_i \mbox{ literal of } C_j\} \cup   \{\bar{x}_ic_j \mid \bar{x}_i \mbox{ literal of } C_j\}.
 \end{eqnarray*}

Observe that in $M$ the vertices $y_i$ are sources, the vertices $z_i$
are sinks, and the vertices $c_j$ are sinks too.  Hence if $(V_1,V_2)$
is a $(\delta^+\geq 1,\delta^-\geq 1)$-bipartite-partition of $M$,
then for each $i\in [n]$ we have $y_i \in V_1$ and $z_i \in V_2$, and
for each $j\in [m]$ we have $c_j\in V_2$.  Consequently, for each
$i\in [n]$, exactly one of the vertices $x_i,\bar{x}_i$ belongs to
$V_1$. Hence if we interpret $x_i\in V_1$ (resp. $x_i\in V_2)$ as
meaning the the variable $v_i$ is true (resp. false), then $M$ has a
$(\delta^+\geq 1,\delta^-\geq 1)$-bipartite-partition if and only if
${\cal F}$ is satisfiable.
\end{proof}

\subsection{ Solving {\sc $(\delta^+\geq 1,\delta^-\geq 1)$-bipartite-partition} for strong digraphs}

The digraph $M$ in the above proof is very far from being strong as it
has many sources and sinks. A natural question is thus to determine
the complexity of $(\delta^+\geq 1,\delta^-\geq
1)$-bipartite-partition problem when the input is restricted to be a
strong digraph. We show below that in this case the problem becomes
solvable in polynomial time.\\


For any digraph $D$ we define the following reduction rule.

\begin{description}
\item[Reduction Rule A:] If for some arc $xy \in A(D)$ we have
  $d^+(x)=d^-(y)=1$ then we reduce $D$ by deleting $x$ and $y$ and
  adding all arcs from $N^-(x)$ to $N^+(y)$ in $D$ (if an arc $uv$ is
  already present we do not add an extra copy. Similarly, if $z\in
  N^-(x)\cap N^+(y)$, then we do not add a loop at $z$.).

In this case, we say that the arc $xy$ got {\bf reduced}.
\end{description}

We call $D'$ a {\bf reduction of } $D$ if $D'$ is obtained from $D$ by
applying the reduction rule one or more times. We first prove the
following lemma.

\begin{lemma} \label{lem1}
If $D'$ is a reduction of $D$, then $D$ has a \Gcol{} if and only if
$D'$ has a \Gcol{}.  Furthermore if $D$ is strong then so is $D'$.
\end{lemma}

\begin{proof}
Let $D$ and $D'$ be defined as in the statement of the lemma.  Clearly
if suffices to prove the lemma when $D'$ was obtained from $D$ by
applying the reduction rule A once, since then the claim follows by
induction. So let $xy$ be the arc that got reduced in $D$ in order to
obtain $D'$.

\medskip
First assume that $D'$ has a \Gcol{} $c$. Now consider the following
two cases.

\2

{\bf Case 1}. There exists $u \in N_D^-(x)$ and $v \in N_D^+(y)$ such
that $c(u)=1$ and $c(v)=2$.  In this case we can assign $c(x)=2$ and
$c(y)=1$ and note that $c$ is now a \Gcol{} of $D$.

\2

{\bf Case 2}. All vertices in $N_D^-(x)$ have colour $2$ or all
vertices in $N_D^+(y)$ have colour $1$.  In this case we can assign
$c(x)=1$ and $c(y)=2$ and note that $c$ is now a \Gcol{} of $D$.

\2

Therefore if $D'$ has a \Gcol{} then so does $D$.

\medskip

Now assume that $D$ has a \Gcol{}, $c$ and consider the following two
cases.

\2

{\bf Case A}. $c(x)=2$. In this case $c(y)=1$ as $y$ does not have any
in-neighbour of colour $1$.  Therefore there must be a vertex in $
N_D^-(x)$ with colour $1$ and a vertex in $ N_D^+(y)$ with colour
$2$. Therefore just restricting the $2$-colouring $c$ to $V(D')$ gives
us a \Gcol{} for $D'$.

\2

{\bf Case B}. $c(x)=1$. In this case $c(y)=2$ as $x$ needs an
out-neighbour of colour $2$.  Now restricting the $2$-colouring $c$ to
$D'$ gives us a \Gcol{} for $D'$.

\2

This completes the proof of the first part of the lemma.

\bigskip

Assume that $D$ is strong. We will show that $D'$ is also strong.  Let
$u,v \in V(D')$ be arbitrary. As $D$ is strong there exists a
$(u,v)$-path, $P$, in $D$.  If $x$ or $y$ belong to $P$, then the arc
$xy$ belongs to $P$ as in $D$ we have $d^+(x)=d^-(y)=1$.  If $x^-$ is
the predecessor of $x$ on $P$ and $y^+$ is the successor of $y$ on
$P$, then we obtain a $(u,v)$-path in $D'$ by deleting $x$ and $y$
from $P$ and adding the arc $x^- y^+$ (which by definition belongs to
$D'$).  As $u$ and $v$ were picked arbitrarily this implies that $D'$
is strong.
\end{proof}

A non-trivial {\bf out-star} (resp. {\bf in-star}) is an out-tree
(resp. in-tree) of depth $1$, that is, it consists of at least two
vertices and the root dominates (resp. is dominated by) all the other
vertices in the tree. An {\bf out-galaxy} (resp. {\bf in-galaxy}) is a
set of vertex-disjoint non-trivial out-stars (resp. in-stars).  A {\bf
  nebula} is a set of vertex-disjoint non-trivial out- or in-stars.

Every nebula has a \Gcol{} : colour with $1$ the sinks and with $2$
the sources. Consequently, if a digraph has a spanning nebula, then it
also has a \Gpart{}.  The following result is proved in
\cite{goncalvesDAM160}.

\begin{lemma}[\cite{goncalvesDAM160}] \label{lem2}
 If $D$ is a strong digraph and $D' \subseteq D$ is a strong
 subdigraph of $D$ of even order, then $D$ has a spanning out-galaxy
 and therefore also a \Gpart{}.
\end{lemma}

\begin{corollary}\label{cor1}
 Let $D$ be a strong digraph and assume that there exists an arc $xy
 \in A(D)$ and two vertex disjoint $(y,x)$-paths in $D$.  Then $D$ has
 a \Gpart{}.
\end{corollary}

\begin{proof}
 Let $P_1$ and $P_2$ be the two vertex-disjoint $(y,x)$-paths in $D$.
 If $P_i$ ($i \in \{1,2\}$) has an even number of vertices, then, by
 Lemma~\ref{lem2}, $D$ has a \Gpart{} as $D\langle V(P_i)\rangle$ is
 strong (as $P_i \cup xy$ is a cycle) and of even order.  So we may
 assume that both $P_1$ and $P_2$ have an odd number of vertices.  Let
 $D' = D\langle V(P_1) \cup V(P_2) \rangle$ and note that $D'$ is
 strong and $|V(D')| = |V(P_1)|+|V(P_2)|-2$ is even, implying that $D$
 has a \Gpart{}, by Lemma~\ref{lem2}.
\end{proof}

\begin{theorem} \label{thmgoodpart}
If $D$ is a strong digraph, then one of the following holds.

\begin{description}
 \item[(a):] $D$ has a \Gpart{}.
 \item[(b):] $D$ is an isolated vertex.
 \item[(c):] $D$ can be reduced.
\end{description}
\end{theorem}

\begin{proof}
Assume the theorem is false and that none of (a), (b) or (c) hold.
That is $|V(D)| \geq 2$, $D$ cannot be reduced and $D$ does not have
a \Gpart{}.

We now build a sequence of handles as follows. Let $D_1=H_1$ be a
shortest cycle in $D$ (which exists since $D$ is strong and not an
isolated vertex) and let $i=1$.  While $V(D_i) \not= V(D)$, let
$H_{i+1}$ be a shortest non-trivial $D_i$-handle and $D_{i+1}=D_i\cup
H_{i+1}$.  Continue this until $V(D_i)=V(D)$ (which is easy seen to be
possible as $D$ is strong).
 
Let $H_p=(x,v_1,v_2, \cdots , v_l, y)$ be the last handle added in the
above process.

If $p=1$, then the shortest cycle in $D$ is a Hamilton cycle, and so
$D$ itself is this cycle. As (a) is false, we note that $D$ is not a
$2$-cycle and letting $xy$ be any arc on the cycle we note that $D$
can be reduced (as $d^+(x)=d^-(y)=1$), a contradiction.  So we may
assume that $p \geq 2$.

As all $D_i$ ($i\in [p]$) are strong and $D$ has no
\Gpart{}, both $|V(D_{p-1})|$ and $|V(D_p)|$ are odd, by
Lemma~\ref{lem2}.  As $|V(D_{p-1})| + l = |V(D_p)|$ we must therefore
have that $l$ is even.  We now prove a number of claims.

\begin{claim}\label{claim1}
$d_D^+(v_1) = 1$ and $d_D^-(v_l)=1$.
\end{claim}

\noindent {\em Proof of Claim~\ref{claim1}.} For the sake of contradiction, assume that
$d_D^+(v_1) > 1$ and let $v_1 z$ be any arc out of $v_1$ different
from $v_1v_2$. Note that $ z \not\in D_{p-1}$ by the minimality of
$H_p$. Also $z \not\in \{v_3,v_4,\ldots,v_l\}$ by the minimality of
$H_p$. As $z$ is also not $v_1$ or $v_2$, $z$ does not exist, a
contradiction.  This proves that $d_D^+(v_1) = 1$.
We can prove $d_D^-(v_l)=1$ analogously. ~\hfill $\Diamond$

\begin{claim}\label{claim2}
If $d_D^+(v_i) = 1$, then $d_D^-(v_{i+1})>1$ for all
$i=1,2\ldots, l-1$.
\end{claim}

\noindent{\em Proof of Claim~\ref{claim2}.} This follows immediately from the fact that
$D$ cannot be reduced. ~\hfill $\Diamond$

\begin{claim}\label{claim3}
If $d_D^-(v_i) > 1$, then $d_D^+(v_{i+1})=1$ for all
$i=2,3,\ldots, l-2$.
\end{claim}

\noindent{\em Proof of Claim~\ref{claim3}.} For the sake of contradiction, assume that
$d_D^-(v_i) > 1$ and $d_D^+(v_{i+1})>1$ for some $i \in \{2,3,\ldots,
l-2\}$.  Let $zv_i$ be any arc in $D$ different from
$v_{i-1}v_i$. Note that $z=v_j$ for some $j \in \{i+1,i+2,\ldots,l\}$
as otherwise there would exist a shorter handle than $H_p$.
Analogously let $v_{i+1}v_k$ be an arc out of $v_{i+1}$ different from
$v_{i+1}v_{i+2}$ and note that $k \in \{1,2,\ldots,i\}$. If $j=i+1$ or
$k=i$ then $v_i v_{i+1} v_i$ is a $2$-cycle, which is also a strong
digraph of even order, a contradiction by Lemma~\ref{lem2}.  So
$j>i+1$ and $k<i$ which implies that we get a contradiction to
Corollary~\ref{cor1} (as $v_iv_{i+1}$ is an arc and
$v_{i+1}v_{i+2}\ldots v_j v_i$ and $v_{i+1} v_k v_{k+1} \ldots v_i$
are vertex-disjoint paths).  This contradiction completes the proof of
Claim~\ref{claim3}. ~\hfill $\Diamond$

\2

By Claim~\ref{claim1}, $d_D^+(v_1)=1$.  By Claim~\ref{claim2}, $d_D^-(v_2)>1$.  By Claim~\ref{claim3},
$d_D^+(v_3)=1$.  By Claim~\ref{claim2}, $d_D^-(v_4)>1$. Continuing this, we note
that $d_D^+(v_i)=1$ for all odd $i<l$ and $d_D^-(v_i)>1$ for all even
$i \leq l$.  However, by Claim~\ref{claim1}, we note that $d_D^-(v_l)=1$, which
is the desired contradiction, as $l$ was even.
\end{proof}

\begin{corollary}
\label{strong+1-1}
We can decide in polynomial time whether a strong digraph has a
\Gpart{}.
\end{corollary}

\begin{proof}
Let $D$ be a strong digraph. We continuously reduce the digraph, $D$,
until it cannot be reduced any more.  Let $D'$ be the resulting
digraph. By Lemma~\ref{lem1}, $D'$ is strong and has a \Gpart{} if and
only if $D$ does.  By Theorem~\ref{thmgoodpart}, $D'$, and therefore
$D$, has a \Gpart{} if and only if $D'$ is not a single vertex.  All
the operations can be done in polynomial time, which completes the
proof.
\end{proof}

\subsection{Classification of ${\cal NP}$-complete instances in terms of their strong component digraph}

The {\bf strong component digraph}, denoted by $SC(D)$, of a digraph
is obtained by contracting every strong component of $D$ to a single
vertex and deleting any parallel arcs obtained in the process.  For
any acyclic digraph $H$, let ${\cal D}^c(H)$ denote the class of all
digraphs $D$ with $SC(D)=H$.

\begin{proposition}
\label{acyclicgoodpart}
A digraph $D$ has a \Gpart{} if and only if it has a spanning
nebula. Furthermore, given any \Gpart{} of $D$ we can produce a
spanning nebula in polynomial time, and vice-versa.
\end{proposition}

\begin{proof}
If $\cal F$ is a spanning nebula of $D$, then we obtain a \Gpart{} by
colouring the root of every out-star by $1$ and the leaves by $2$ and
the root of every in-star by $2$ and its leaves by $1$. Suppose
conversely that $(V_1,V_2)$ is a \Gpart{} of $D$, where $V_i$ denotes
the vertices of colour $i$ and let $D'$ be the spanning subdigraph of
$D$ induced by the arcs from $V_1$ to $V_2$. Clearly it suffices to
prove that $D'$ has a spanning nebula. We prove this by induction on
the number of vertices. If $D'$ has just two vertices $x,y$, then this
is clear so assume $|V(D')|\geq 3$. Let $v_1v_2$ be an arc of $D'$
with $v_i\in V_i$, $i=1,2$. If $(V_1\setminus\{v_1\},V_2\setminus
\{v_2\})$ is a \Gpart{} of $D'-\{v_1,v_2\}$, then we are done by
induction, so we may assume that either $v_1$ is the unique
in-neighbour of some non-empty set $V'_2\subseteq V_2$, or $v_2$ is
the unique out-neighbour of some non-empty set $V'_1\subseteq V_1$. We
choose $V'_1,V'_2$ to be maximal with the given property. By the
assumption that $(V_1\setminus \{v_1\},V_2\setminus\{v_2\})$ is not a
\Gpart{} of $D'-\{v_1,v_2\}$, we have $V'_i\neq \{v_i\}$ for $i=1$ or
$i=2$. Without loss of generality, we have
$V'_2\setminus\{v_2\}\neq\emptyset$. Now it is easy to see that
$(V_1\setminus \{v_1\},V_2\setminus V'_2)$ is a \Gpart{} of
$D'-(\{v_1\}\cup{}V'_2)$ and we are done by induction. The process
above clearly yields a polynomial-time algorithm for producing the
spanning nebula.
\end{proof}


Let $B_r^+$ be an $r$-out-branching.  Consider the following procedure
that produces an out-galaxy: Let $v$ be a leaf at maximum depth, let
$v'$ be its parent and let $T^+_{v'}$ be the out-tree rooted at $v'$
in $B^+_r$. Then $T^+_{v'}$ is an out-star. Remove this from $B^+_r$
and continue recursively until either no vertex remains or only the
root $r$ remains. In first case, we say that $B^+_r$ is {\bf winning}
and in the second case that $B^+_r$ is {\bf losing}. Observe that if
the root $r$ dominates a leaf in $B_r^+$, then $B^+_r$ is winning.
Similarly, an in-branching is {\bf winning} (resp. {\bf losing}) if
its converse is winning (resp. losing). It is easy to check the
following.

\begin{proposition}\label{prop:winning}
Let $B_r^+$ be an $r$-out-branching.
\begin{itemize}
\item If $B_r^+$ is winning, then it has a spanning out-galaxy, and so
  a \Gpart{}.
\item If $B_r^+$ is losing, then $B^+_r-r$ has a spanning out-galaxy,
  and so a \Gpart{}.
\end{itemize}
\end{proposition}

We sometimes use this proposition without explicitly referring to it.

\begin{theorem} \label{thm2}
Let $H$ be any connected acyclic digraph of order at least $2$. The
following now holds.
\begin{itemize}
\item[(a)] If $H$ has a \Gpart{}, then all digraphs in ${\cal D}^c(H)$
  have a \Gpart{} and we can produce such a partition in polynomial
  time.
\item[(b)] If $H$ has no \Gpart{}, then it is ${\cal NP}$-complete to
  decide whether a digraph in ${\cal D}^c(H)$ has a \Gpart{}.
\end{itemize}
\end{theorem}

\begin{proof}
We first prove (a). Let $H$ be an acyclic digraph which has a \Gpart{}
and let ${\cal F}$ be a spanning nebula of $H$ (by
Proposition~\ref{acyclicgoodpart}). Let $D$ be any digraph such that
$SC(D)=H$. We prove by induction on the number of stars in $\cal N$
that $D$ has a \Gpart{}.

Suppose first that $\cal N$ consists of one star. Since a digraph has
a \Gpart{} if and only if its converse (obtained by reversing all
arcs) has one, we may assume w.l.o.g. that $\cal N$ consists of an
out-star $S^+_r$ with root $r$ and leaves $s_1,s_2,\ldots{},s_k$. Let
$R,S_1,\ldots{},S_k$ be the strong components of $D$ that correspond
to these vertices. Fix an arc $uv$ such that $u\in R,v\in S_1$. As all
of the digraphs $R,S_1,S_2,\ldots{},S_k$ are strong, they all have an
out-branching rooted at any prescribed vertex. In particular this
implies that $D'=D-(S_1-v)$ has an out-branching $B'$ rooted at
$u$. Since its root $u$ is adjacent to one of its leaves $v$, the
out-branching $B'$ is winning. Hence, by
Proposition~\ref{prop:winning}, $B'$ has a spanning out-galaxy, and so
$D'$ has a \Gpart{} $(V'_1, V'_2)$. Observe that necessarily $v\in
V'_2$ because it is a sink in $D'$.  Let $B''$ be an in-branching of
$S_1$ rooted at $v$.  If $B''$ is losing, then $S_1-v$ has a \Gpart{}
$(V''_1, V''_2)$, and $(V'_1\cup V''_1, V'_2\cup V''_2)$ is a \Gpart{}
of $D$.  If $B''$ is winning, then $B''$ has a \Gpart{} $(V''_1,
V''_2)$. In addition $v\in V''_2$, because it is a sink in $B''$.
Thus $(V'_1\cup V''_1, V'_2\cup V''_2)$ is a \Gpart{} of $D$.

Assume now that $\cal N$ has more than one star. Let $S$ be such a
star (out- or in-). Then it follows from the proof above that the
subdigraph of $D$ induced by the vertices of those strong components
that are contracted into $S$ in $SC(D)$ has a \Gpart{}. Now that
partition can be combined with any \Gpart{} of the digraph induced by
the remaining strong components, the existence of which follows by
induction. This completes the proof of (a).

\bigskip

We proceed to prove (b). Let $H$ be a connected acyclic digraph on at
least two vertices which has no \Gpart{}.  We construct a maximal
induced subdigraph $H'$ from $H$ which has a \Gpart{} as follows.  Let
$S$ contain all sinks in $H$ and let $H'$ be induced by $S \cup
N^-(S)$. Clearly $H'$ has a \Gpart{} (by letting the vertices in $S$
have colour $2$ and the vertices in $N^-(S)$ have colour $1$).  Now
repeatedly add vertices or a set of vertices to $H'$ such that $H'$
has a \Gpart{}. When no more vertices can be added we have our desired
$H'$.  Let $X$ be the set of vertices not in $H'$ and $X'=N^+(X)$. We
fix a \Gcol{} $c'$ of $H'$.


Now let $\cal F$ be an instance of 3-SAT with variables
$x_1,x_2,\ldots{},x_n$ and clauses $C_1,C_2,\ldots{},C_m$. We may
assume that $\cal F$ cannot be satisfied by setting all variables true
or all variables false. Form the digraph $W=W(\cal F)$ as follows: the
vertex set of $W$ is
$a,b,c_1,c_2,\ldots{},c_m,v_1,v_2,\ldots{},v_n$. The arc set consists
of all arcs from $a$ to $\{c_1,c_2,\ldots{},c_m\}$, all arcs from
$\{c_1,c_2,\ldots{},c_m\}$ to $b$, all arcs from $b$ to
$\{v_1,v_2,\ldots{},v_n\}$, all arcs from $\{v_1,v_2,\ldots{},v_n\}$
to $a$ and the following arcs between $\{c_1,c_2,\ldots{},c_m\}$ and
$\{v_1,v_2,\ldots{},v_n\}$: For each $j\in [m]$, and $i\in [n]$, if $C_j$
contains the literal $x_i$ we add the arc $c_jv_i$ to $W$, and if $C_j$ contains the literal $\bar{x}_i$ we add the arc $v_ic_j$ to $W$. 

\begin{claim}\label{claim4}
The digraph $W$ has a \Gcol{} $c$ where $c(a)=2$ and $c(b)=1$
if and only if $\cal F$ is satisfiable.
\end{claim}

\noindent{}{\em Proof of Claim~\ref{claim4}.} 
Assume first that the digraph $W$ has a \Gcol{} $c$ where $c(a)=2$ and $c(b)=1$.
Let $\phi$ be the truth assignment defined by $\phi(x_i)=true$ if $c(v_i)=2$ and $\phi(x_i)=false$ otherwise. 
We claim that $\phi$ satisfies $\cal F$: For
each $j\in [m]$ consider the vertex $c_j$. 
If $c(c_j) =1$, then $c_j$ has an out-neighbour
$v_q$ coloured $2$, because $c$ is a \Gcol{}.  But by construction $x_q\in C_j$, and $\phi(x_q)=true$ by definition.
Hence the clause $C_j$ is satisfied.
 Similarly, if $c(c_j) = 2$, then it has an in-neighbour  $v_p$ coloured $1$, and $\bar{x}_p \in C_j$ and $\phi(x_q)=false$.
 So the clause $C_j$ is satisfied.

 Conversely, given a truth assignment $\phi$ which satisfies ${\cal F}$, we start by
colouring $v_i$, $i\in [n]$ by 2 if $x_i$ is true and 1
otherwise. Since $\phi$ satisfies all clauses it is easy to check that we
can extend this colouring to all vertices of
$\{c_1,c_2,\ldots{},c_m\}$. As $\phi$ sets at least one variable true
and at least one false (by our assumption on $\cal F$), we can finish
the colouring by colouring $a$ by colour 2 and $b$ by colour 1. This gives the desired \Gcol{}  and
completes the proof of the claim. ~\hfill $\Diamond$

\medskip

We will now show how to form a digraph in ${\cal D}^c(H)$ which has a
\Gpart{} if and only if $\cal F$ is satisfiable. Fix a vertex $x\in X$
and an out-neighbour $y\in X'$ of $x$.  Construct the digraph $D$ from
$H$ and $W$ as follows: For every vertex of $u\in X\setminus \{x\}$,
we add three new (private) vertices $u_1,u_2,u_3$ and the arcs of the
$4$-cycle $(u_1,u_2,u_3,u,u_1)$.  Replace the vertex $y$ by a copy of
$W$ where we let every arc into $y$ in $H$ enter the vertex $a$
(eg. $xy$ becomes $xa$) and let every arc out of $y$ be incident with
$b$.

Suppose first that $\cal F$ is satisfiable. By Claim~\ref{claim4}, there exists a \Gcol{} $c$ of $W$ with $c(a)=2,c(b)=1$.  This can easily be extended
to a \Gcol{} of $D$ by letting $c(x)=1$, colouring each of the
private 4-cycles $(u_1,u_2, u_3,u,u_1)$ as
$c(u_1)=c(u_3)=1,c(u)=c(u_2)=2$ and extending this colouring to the
remaining vertices of $H'-y$ using $c'$ above.

Suppose now that $D$ has a \Gpart{} and let $c^*$ be the associated \Gcol.  By the claim above, it suffices to prove
that we must have $c^*(a)=2$ and $c^*(b)=1$.

For the sake of contradiction assume that $c^*(b)=2$.  In this case if
we restrict $c^*$ to $V(H'-y)$ and assign $c^*(x)=1$ and $c^*(y)=2$ we
get a \Gcol\ of $H' \cup \{x\}$ contradicting the fact that
$H'$ was maximal. Therefore $c^*(b)=1$. Now, for the sake of
contradiction assume that $c^*(a)=1$.  In this case if we restrict
$c^*$ to $V(H'-y) \cup \{x\}$ and assign $c^*(y)=1$ we get a \Gcol\ of $H' \cup \{x\}$ contradicting the fact that $H'$ was
maximal.  Therefore $c^*(a)=2$ and $c^*(b)=1$ and the proof is
complete.
 \end{proof}

\section{$(\delta^+\geq k_1,\delta^-\geq k_2)$-bipartite-partitions when $k_1+k_2\geq 3$}\label{sec:5}

\begin{theorem} \label{thmk1}
Let $k_1 \geq 2$ be an integer.
It is ${\cal NP}$-complete to decide whether a given strong digraph
$D$ has a $(\delta^+\geq k_1,\delta^-\geq 1)$-bipartite-partition.
\end{theorem}

\begin{proof}
Reduction from 3-SAT.

Let $Q$ be the digraph whose vertex set is the disjoint union of two sets $W,Z$ of size $k_1-1$, and $\{v,\bar{v},y,z\}$ and with arc set
$\{yv,y\bar{v}\}  \cup \bigcup_{w\in W} \{wy, yw\}  \cup \bigcup_{z\in Z} \{vz,\bar{v}z,zy\}$. See Figure~\ref{figureQ}.

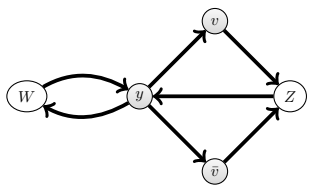
\begin{figure}[h!]
\begin{center}
\tikzstyle{vertexX}=[circle,draw, top color=gray!5, bottom color=gray!30, minimum size=16pt, scale=0.6, inner sep=0.5pt]
\tikzstyle{vertexZ}=[circle,draw, top color=white!5, bottom color=white!30, minimum size=4pt, scale=0.6, inner sep=0.5pt]
\tikzstyle{vertexBIG}=[ellipse, draw, scale=0.6, inner sep=3.5pt]
\begin{tikzpicture}[scale=0.5]

 \node (w) at (1.0,3.0) [vertexBIG] {$W$};
 \node (y) at (4.0,3.0) [vertexX] {$y$};
 \node (v1) at (6.0,5.0) [vertexX] {$v$};
 \node (v2) at (6.0,1.0) [vertexX] {$\bar{v}$};
 \node (z) at (8.0,3.0) [vertexBIG] {$Z$};

    \draw [->, line width=0.05cm] (y) -- (v1);
    \draw [->, line width=0.05cm] (y) -- (v2);
    \draw [->, line width=0.05cm] (v1) -- (z);
    \draw [->, line width=0.05cm] (v2) -- (z);
    \draw [->, line width=0.05cm] (z) -- (y);
    \draw [->, line width=0.05cm] (w) to [out=30, in=150] (y);
    \draw [->, line width=0.05cm] (y) to [out=210, in=330] (w);

\end{tikzpicture}
\end{center}
\caption{The gadget $Q$.}
\label{figureQ}
\end{figure}

Let $\cal F$ be an instance of 3-SAT with variables
$x_1,x_2,\ldots{},x_n$ and clauses $C_1,C_2,\ldots{},C_m$. By adding
extra clauses not affecting the truth value of the original ones if
necessary, we may assume that every variable $x_i$ appears in some
clause as the literal $x_i$ and in another clause as the literal
$\bar{x}_i$.  Form the digraph $D=D(\cal F)$ as follows: take $n$
disjoint copies $Q_1,Q_2,\ldots{},Q_n$ of $Q$ and denote the sets corresponding to $W,Z$  in $Q_i$ by $W_i, Z_i$ respectively and the vertices
of $Q_i$ corresponding to  $v,\bar{v},y$ by $v_i,\bar{v}_i,y_i$ respectively. The vertices $v_i$ and
  $\bar{v_i}$ will correspond to the variable $x_i$: $v_i$ to the
  literal $x_i$ and $\bar{v}_i$ to the literal $\bar{x}_i$. Add $m$
vertices $c_1,c_2,\ldots{},c_m$, where $c_j$ corresponds to clause
$C_j$, $j\in [m]$. Add the arcs of a directed cycle with vertex set $\bigcup_{i\in[n]} Z_i$,
and the arc $c_jy_1$ for all $j\in [m]$. Finally,
for each $j\in [m]$ we add three arcs from the vertices corresponding
to the literals of $C_j$ to the vertices $c_j$. E.g. if
$C_j=(\bar{x}_1\vee \bar{x}_2\vee x_3)$ then we add the arcs
$\bar{v}_1c_j,\bar{v}_2c_j,v_3c_j$. It is easy to check that $D$ is
strong.

Assume that $(V_1,V_2)$ is a  $(\delta^+\geq k_1,\delta^-\geq 1)$-bipartite-partition  of $D$.
 Since the out-degree of each $c_j$, $j\in [m]$ is $1$,
these vertices must belong to $V_2$. Similarly, $W_i\subset
V_2$, for all $i\in [n]$. This  implies that $y_i\in V_1$
and thus $Z_i\subset V_2$ for every $i\in [n]$.
Those two facts imply that exactly one of $v_i,\bar{v}_i$
belong to $V_1$ and the other belongs to $V_2$.  
Let $\phi$ be the truth assignment defined by $\phi(x_i)=true$ if $v_i\in V_1$ and $\phi(x_i)=false$ otherwise.
One easily sees that $\phi$ satisfies $\cal F$ as every vertex $c_j$ has an in-neighbour in $V_1$ which is a vertex corresponding to a 
literal of $C_j$ which is then assigned $true$ by $\phi$.

Reciprocally, assume that there is a truth assignment $\phi$ satisfying $\cal F$.
 Let $(V_1,V_2)$ be the $2$-partition of $D$ defined by 
\begin{align*}
 V_1 &=\{y_i \mid i\in [n]\}\cup \{v_i \mid \phi(x_i)=true\} \cup  \{\bar{v}_i \mid \phi(x_i)=false\}, \mbox{and}  \\ 
 V_2 & =\{c_j \mid j\in [m]\} \cup \bigcup_{i=1}^n (W_i\cup Z_i) \cup \{v_i \mid \phi(x_i)=false\} \cup  \{\bar{v}_i \mid \phi(x_i)=true\} .
 \end{align*}
 One easily checks that $(V_1,V_2)$ is a  $(\delta^+\geq k_1,\delta^-\geq 1)$-bipartite-partition  of $D$.
In particular, since $x_i$ and $\bar{x}_i$ belong to at least one clause, $v_i$ and
  $\bar{v}_i$ have each at least one out-neighbour in $\{c_1,\dots
  ,c_m\}$, which is a subset of $V_2$.
  \end{proof}

\begin{theorem} \label{thm22}
It is ${\cal NP}$-complete to decide whether a given strong digraph
$D$ has a $(\delta^+\geq 2,\delta^-\geq 2)$-bipartite-partition.
\end{theorem}

\begin{proof}
The proof is a reduction from 3-SAT, which is very similar to the one of Theorem~\ref{thmk1}.

Let $Q'$ be the digraph with vertex set $\{w, y, y', v,\bar{v},z\}$ and with arc set
$$\{y'y, yw,y'w,wy', yv,y\bar{v}, y'v,y'\bar{v}, y'z, zy, vz,\bar{v}z\}.$$ 
See Figure~\ref{figureQ'}.

\begin{figure}[h!]
\begin{center}
\tikzstyle{vertexX}=[circle,draw, top color=gray!5, bottom color=gray!30, minimum size=16pt, scale=0.6, inner sep=0.5pt]
\tikzstyle{vertexZ}=[circle,draw, top color=white!5, bottom color=white!30, minimum size=4pt, scale=0.6, inner sep=0.5pt]
\tikzstyle{vertexBIG}=[ellipse, draw, scale=0.6, inner sep=3.5pt]
\begin{tikzpicture}[scale=0.5]

  \node (w) at (1.0,3.0) [vertexX] {$w$};
 \node (y) at (4.0,2.0) [vertexX] {$y$};
  \node (y') at (4.0,4.0) [vertexX] {$y'$};
 \node (v1) at (7.0,5.5) [vertexX] {$v$};
 \node (v2) at (7.0,0.5) [vertexX] {$\bar{v}$};
 \node (z) at (9.0,3.0) [vertexX] {$z$};

    \draw [->, line width=0.05cm] (y) -- (v1);
    \draw [->, line width=0.05cm] (y) -- (v2);
    \draw [->, line width=0.05cm] (y') -- (v1);
    \draw [->, line width=0.05cm] (y') -- (v2);
    \draw [->, line width=0.05cm] (y') -- (z);
    \draw [->, line width=0.05cm] (v1) -- (z);
    \draw [->, line width=0.05cm] (v2) -- (z);
    \draw [->, line width=0.05cm] (z) -- (y);
    \draw [->, line width=0.05cm] (y') -- (y);
    \draw [->, line width=0.05cm] (y) -- (w);
    \draw [->, line width=0.05cm] (y') -- (w);
    \draw [->, line width=0.05cm] (w) to [out=60, in=150] (y');

\end{tikzpicture}
\end{center}
\caption{The gadget $Q'$.}
\label{figureQ'}
\end{figure}
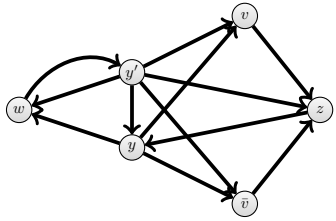

Let $\cal F$ be an instance of 3-SAT with variables
$x_1,x_2,\ldots{},x_n$ and clauses $C_1,C_2,\ldots{},C_m$. By adding
extra clauses not affecting the truth value of the original ones if
necessary, we may assume that every variable $x_i$ appears in some
clause as the literal $x_i$ and in another clause as the literal
$\bar{x}_i$.  Form the digraph $D=D(\cal F)$ as follows: take $n$
disjoint copies $Q'_1,Q'_2,\ldots{},Q'_n$ of $Q'$ and denote the vertices of $Q_i$
corresponding to  $w,y,y',v,\bar{v},z$  by $w_i, y_i, y'_i,v_i,\bar{v}_i,z_i$ respectively. The vertices $v_i$ and
  $\bar{v_i}$ will correspond to the variable $x_i$: $v_i$ to the
  literal $x_i$ and $\bar{v}_i$ to the literal $\bar{x}_i$. Add $m$
vertices $c_1,c_2,\ldots{},c_m$, where $c_j$ corresponds to clause
$C_j$, $j\in [m]$. Add the arcs of the directed cycle $z_1z_2 \ldots z_nz_1$,
and the arc $c_jy_1$ for all $j\in [m]$. Finally,
for each $j\in [m]$ we add three arcs from the vertices corresponding
to the literals of $C_j$ to the vertices $c_j$. E.g. if
$C_j=(\bar{x}_1\vee \bar{x}_2\vee x_3)$ then we add the arcs
$\bar{v}_1c_j,\bar{v}_2c_j,v_3c_j$. It is easy to check that $D$ is
strong.

Assume that $(V_1,V_2)$ is a  $(\delta^+\geq 2,\delta^-\geq 2)$-bipartite-partition  of $D$.
 Since the out-degree of each $c_j$, $j\in [m]$ is $1$,
these vertices must belong to $V_2$. Similarly, $w_i\in
V_2$, for all $i\in [n]$. This  implies that $\{y_i, y'_i\}\subset V_1$
and thus $z_i\in V_2$ for every $i\in [n]$.
Those two facts imply that exactly one of $v_i,\bar{v}_i$
belong to $V_1$ and the other belongs to $V_2$.  
Let $\phi$ be the truth assignment defined by $\phi(x_i)=true$ if $v_i\in V_1$ and $\phi(x_i)=false$ otherwise.
One easily sees that $\phi$ satisfies $\cal F$ as every vertex $c_j$ has an in-neighbour in $V_1$ which is a vertex corresponding to a 
literal of $C_j$ which is then assigned $true$ by $\phi$.

Reciprocally, assume that there is a truth assignment $\phi$ satisfying $\cal F$.
 Let $(V_1,V_2)$ be the $2$-partition of $D$ defined by 
\begin{align*}
 V_1 &=\bigcup_{i=1}^n \{y_i, y'_i\} \cup \{v_i \mid \phi(x_i)=true\} \cup  \{\bar{v}_i \mid \phi(x_i)=false\}, \mbox{and}  \\ 
 V_2 & =\{c_j \mid j\in [m]\} \cup \bigcup_{i=1}^n \{w_i, z_i\} \cup \{v_i \mid \phi(x_i)=false\} \cup  \{\bar{v}_i \mid \phi(x_i)=true\} .
 \end{align*}
 One easily checks that $(V_1,V_2)$ is a  $(\delta^+\geq 2,\delta^-\geq 2)$-bipartite-partition  of $D$.
In particular, since $x_i$ and $\bar{x}_i$ belong to at least one clause, $v_i$ and
  $\bar{v}_i$ have each at least one out-neighbour in $\{c_1,\dots
  ,c_m\}$, which is a subset of $V_2$.
  \end{proof}

\begin{corollary}
Let $k_1,k_2 \geq 1$ be positive integers such that $k_1+k_2\geq 3$.  It is
${\cal NP}$-complete to decide whether a given strong digraph $D$ has
a $(\delta^+\geq k_1,\delta^-\geq k_2)$-bipartite-partition.
\end{corollary}

\begin{proof}
Without loss of generality we may assume that $k_1 \geq k_2$
(otherwise swap $k_1$ and $k_2$ and reverse all arcs).

We prove the result by induction on $k_1+k_2$.
If $k_2=1$, then we have the result by Theorem~\ref{thmk1}, and if  $k_1=k_2=2$, we have the result by Theorem~\ref{thm22}.

\medskip

Assume now that $k_1+k_2\geq 5$ and $k_2\geq 2$.
We give a reduction from {\sc $(\delta^+\geq k_1-1,\delta^-\geq
  k_2-1)$-bipartite-partition} which is ${\cal NP}$-complete  by the induction hypothesis.
 Let $D$ be a
digraph. We construct $D'$ from $D$ by adding two vertices $x_1,x_2$ and all arcs from $V(D)$ to $x_2$, all arcs from
$x_1$ to $V(D)$ and the two arcs $x_1x_2, x_2x_1$.  Clearly $D'$ is strong.
  
For any $(\delta^+\geq k_1,\delta^-\geq k_2)$-bipartite-partition
$(V'_1,V'_2)$ of $D'$,  $x_2$ is in  $V'_2$ because it has t
out-degree $1$, and $x_1\in V'_1$ because it has in-degree $1$. Therefore $(V'_1\setminus  \{x_1\},
V'_2\setminus \{x_2\})$ is a $(\delta^+\geq k_1-1,\delta^-\geq
k_2-1)$-bipartite-partition of $D$.  Reciprocally, if there is
$(\delta^+\geq k_1-1,\delta^-\geq k_2-1)$-bipartite-partition $(V_1,V_2)$ of
$D$, then $(V_1\cup \{x_1\}, V_2\cup \{x_2\})$ is a $(\delta^+\geq
k_1,\delta^-\geq k_2)$-bipartite-partition of $D'$.

Hence $D'$ has a $(\delta^+\geq k_1,\delta^-\geq
k_2)$-bipartite-partition if and only if $D$ has a $(\delta^+\geq
k_1-1,\delta^-\geq k_2-1)$-bipartite-partition.
\end{proof}

\section{Strong $2$-partitions}\label{strongbip}

Recall that a strong $2$-partition of a digraph $D$ is a partition
$(V_1,V_2)$ such that $B_D(V_1,V_2)$ is strong.

\begin{theorem} \label{thm1}
 For every $r>0$, there exists an $r$-strong eulerian digraph $D$
 which has no strong $2$-partition (that is, $D$ has no spanning
 strong bipartite subdigraph).
\end{theorem}

\newcommand{\SetX}[1]{\mbox{ } $\cdots$ #1 $\cdots$ \mbox{ }}
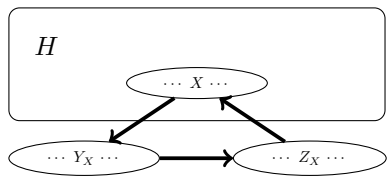
\begin{figure}[h!]
\begin{center}
\tikzstyle{vertexX}=[circle,draw, top color=gray!5, bottom color=gray!30, minimum size=16pt, scale=0.6, inner sep=0.5pt]
\tikzstyle{vertexZ}=[circle,draw, top color=white!5, bottom color=white!30, minimum size=4pt, scale=0.6, inner sep=0.5pt]
\tikzstyle{vertexBIG}=[ellipse, draw, scale=0.6, inner sep=3.5pt]
\begin{tikzpicture}[scale=0.5]

 \draw (1,4) node {$H$};
 \node (x) at (5.0,3.0) [vertexBIG] {\SetX{$X$}};
 \node (y) at (2.0,1.0) [vertexBIG] {\SetX{$Y_X$}};
 \node (z) at (8.0,1.0) [vertexBIG] {\SetX{$Z_X$}};
 \draw [rounded corners] (0.0,2.0) rectangle (10,5);
    \draw [->, line width=0.05cm] (x) -- (y);
    \draw [->, line width=0.05cm] (y) -- (z);
    \draw [->, line width=0.05cm] (z) -- (x);
\end{tikzpicture}
\end{center}
\caption{Adding the gadget $G_r(X)$.}
\label{figure1}
\end{figure}

\begin{proof}
Let $r>0$ be an arbitrary integer, let $H$ be an arbitrary digraph on
at least $r$ vertices and let $X \subseteq V(H)$ be a subset of size
$r$.  Let $D$ be the digraph that we obtain from $H$ and $X$ by adding
two new vertex sets $Y_X$ and $Z_X$ of size $r$ to $H$ and all arcs
from $X$ to $Y_X$, all arcs from $Y_X$ to $Z_X$ and all arcs from
$Z_X$ to $X$.  The digraph induced by $X\cup Y_X\cup Z_X$ is the
gadget $G_r(X)$ and we say that the digraph $D$ is obtained from $H$
by adding the gadget $G_r(X)$ to $H$ (see Figure~\ref{figure1}).

Now the following holds:
\begin{enumerate}
  \item[(1)] If $(V_1,V_2)$ is a strong $2$-partition of $D$, then the
    vertices of $X$ cannot all belong to the same set $V_i$.
\item[(2)] For every set $S$ of at most $r-1$ vertices in $D$, we
  have that $D\langle (X \cup Y_X \cup Z_X) \setminus S\rangle$ is strong.
\end{enumerate}

Property (2) follows from the fact that the gadget $G_r(X)$ is clearly
$r$-strong. To prove (1) assume that $(V_1,V_2)$ is a strong
$2$-partition, and without loss of generality assume that $X \subseteq
V_1$.  In this case, $Y_X \subseteq V_2$ as if $y \in Y_X \cap V_1$
then $y$ has no arc into it in $B_D(V_1,V_2)$, a contradiction.
Analogously $Z_Y \subseteq V_1$ (as $Y_Z \subseteq V_2$) and $Z_Y
\subseteq V_2$ (as $X \subseteq V_1$), a contradiction.\\

Now let $U$ be any digraph on $2r-1$ vertices such that
$d^+_U(v)=d^-_U(v)$ for all vertices $v$ of $U$ (in particular, $U$
could be just a stable set). Construct $D'$ from $U$ by adding a
gadget $G_r(X)$ for each of the ${2r-1\choose r}$ subsets $X$ of $r$
vertices of $U$. By property (2) of the gadget $G_r(X)$, $D'$ is
$r$-strong and it is easy to check that is is eulerian. Furthermore,
$D'$ cannot have a strong $2$-partition, because in any $2$-partition
$(V_1,V_2)$ of $V(D')$ there will be a set, $X \subseteq V(U)$, of
size $r$ belonging to the same set $V_i$, contradicting property (1)
of the gadget $G_r(X)$.
\end{proof}

\begin{theorem}
\label{ssbseulerianNPC}
 For every fixed positive integer $r\geq 3$, it is ${\cal NP}$-complete to
 decide whether an $r$-strong eulerian digraph has a strong $2$-partition.
\end{theorem}

\begin{proof}
 We prove the result by reduction from 2-colourability of
  $r$-uniform hypergraphs. In this problem, being given an $r$-uniform
  hypergraph, we want to colour its ground set with two colours such
  that no hyperedge is monochromatic. It is known that this problem is
  ${\cal NP}$-complete for $r\ge 3$~\cite{lovaszUM1973}, even restricted to
  connected hypergraphs. So let $\cal H$ be a connected hypergraph on
  ground set $V$ and with hyperedges $\{X_1,\dots ,X_m\}$. We
  construct from $\cal H$ the digraph $D_{\cal H}$ by adding to $V$
  the gadget $G_r(X_i)$ for $i=1,\dots ,m$. By Property~(2) of the
  proof of Theorem~\ref{thm1} and as ${\cal H}$ is connected, $D_{\cal
    H}$ is an $r$-strong eulerian digraph. Using Property~(1) it is
  straightforward to check that $D_{\cal H}$ has a strong $2$-partition
  if, and only if, ${\cal H}$ admits a 2-colouring. 
  \end{proof}

The above proof also works for finding spanning bipartite subgraphs
with semi-degree at least $1$, so we obtain the following.

\begin{theorem}
For every integer $r\geq 3$, it is ${\cal NP}$-complete to decide whether an
$r$-strong eulerian digraph contains a spanning bipartite digraph with
minimum semi-degree at least 1.
\end{theorem}

\section{Remarks and open questions}

We looked at some natural properties of the spanning bipartite
subdigraphs induced by a $2$-partition. We list some further results
and open problems in that field.

\begin{theorem}[\cite{bangTCS526}]
It is ${\cal NP}$-complete to decide whether a digraph has a cycle factor in
which all cycles are even.
\end{theorem}

\begin{corollary}
It is ${\cal NP}$-complete to decide whether a digraph $D$ has a $2$-partition
$(V_1,V_2)$ such that the bipartite digraph $B_D(V_1,V_2)$ has a
cycle-factor.
\end{corollary}

A {\bf total dominating set} in a graph $G=(V,E)$ is a set of vertices
$X\subseteq V$ such that every vertex of $V$ has a neighbour in $X$.
\begin{theorem}[\cite{heggernesNJC5}]
  It is ${\cal NP}$-complete to decide whether a graph $G$ has a
  $2$-partition $(V_1,V_2)$ so that $V_i$ is a total dominating set of
  $G$ for $i=1,2$.
\end{theorem}

This directly implies the following.

\begin{corollary}
  It is ${\cal NP}$-complete to decide whether a symmetric digraph $D$
  has a $2$-partition $(V_1,V_2)$ such that
  $\delta^+(\induce{D}{V_i})\geq 1$ for $i=1,2$ and
  $\delta^+(B_D(V_1,V_2))\geq 1$.
\end{corollary}

For any digraph $D$, it is easy to obtain a $2$-partition of $D$ such
that the bipartite digraph $B_D(V_1,V_2)$ is an eulerian
digraph. Indeed the $2$-partition $(V(D),\emptyset)$ produces a
corresponding bipartite digraph with no arcs which is then eulerian.
On the other hand, if we ask for a bipartite eulerian subdigraph with
minimum semi-degree at least $1$, a slight variation in the proof of
Theorem~\ref{ssbseulerianNPC} gives the following result.

\begin{theorem}
It is ${\cal NP}$-complete to decide whether a digraph $D$ has a
$2$-partition $(V_1,V_2)$ such that $B_D(V_1,V_2)$ is a bipartite
eulerian digraph with minimum semi-degree at least $1$.
\end{theorem}

\begin{proof}
We use the same reduction as in the proof of
Theorem~\ref{ssbseulerianNPC} and the same gadget as in the proof of
Theorem~\ref{thm1}. From a hypergraph ${\cal H}$ with hyperedges
$X_1,\dots ,X_m$ we construct the digraph $D_{\cal H}$.  Using the
same arguments as in the proof of Theorem~\ref{thm1}, it is easy to
see that if $D_{\cal H}$ admits a $2$-partition $(V_1,V_2)$ such that
$B_D(V_1,V_2)$ is a bipartite eulerian subdigraph of $D_{\cal H}$ with
minimum semi-degree at least $1$ then no hyperedge of $\cal H$ is
totally contained in a part $V_i$ what means that $\cal H$ is
2-colourable. Conversely if $\cal H$ is 2-colourable, it is possible
to obtain a partition of $D_{\cal H}$ whose arcs going across form a
bipartite eulerian subdigraph of $D_{\cal H}$ with minimum semi-degree
at least $1$. Indeed if an hyperedge $X_i$ contains $p$ vertices of
colour 1 and $q$ vertices of colour 2, then we colour $p$ vertices of
both $Y_{X_i}$ and $Z_{X_i}$ by colour $1$ and we colour the remaining
$q$ vertices of $Y_{X_i}$ and $Z_{X_i}$ by colour $2$.  It is easy to
see now that this partition of $G_r(X_i)$ produces a spanning eulerian
subdigraph of $G_r(X_i)$, and therefore also for $D_{\cal H}$.
\end{proof}

However if we just ask for a non empty bipartite eulerian subdigraph
of a digraph, we obtain the following question.

\begin{question}
What is the complexity of deciding whether a digraph $D$ has a
$2$-partition $(V_1,V_2)$ such that $B_D(V_1,V_2)$ is an eulerian digraph
with at least one arc?
\end{question}

Notice that if we restrict ourselves to the eulerian instances, this latter
question is equivalent to the following one.

\begin{question}
What is the complexity of deciding whether an eulerian digraph $D$ has
a $2$-partition $(V_1,V_2)$ such that $\induce{D}{V_i}$ is eulerian and
non empty for $i=1,2$ ?
\end{question}

Corollary~\ref{strong+1-1} asserts that we can decide in
  polynomial time whether a strong digraph has a \Gpart{}. On the
  other hand a consequence of Claim~\ref{claim4} is that this problem
  becomes ${\cal NP}$-complete if we fix the colour of two vertices. A
  slight modification in the proof of Claim~\ref{claim4} shows that
  it is also the case if we only fix the colour of one vertex. More
  precisely we look at the following problem.

\begin{problem}[\sc $(\delta^+\geq 1,\delta^-\geq 1)$-bipartite-partition-with-a-fixed-vertex]~\\
\underline{Input}: A digraph $D$, a vertex $x$ of $D$ and a colour
$i\in \{1,2\}$.\\ 
\underline{Question}: Does $D$ admit a $(\delta^+\geq 1,\delta^-\geq 1)$-bipartite-partition
$(V_1,V_2)$ such that $x\in V_i$?
\end{problem}

\begin{theorem}
{\sc $(\delta^+\geq 1,\delta^-\geq
  1)$-bipartite-partition-with-a-fixed-vertex} is ${\cal NP}$-complete
even when restricted to strong digraphs.
\end{theorem}
\begin{proof}  
We use the same reduction from 3-SAT than in the proof of
Claim~\ref{claim4} and the same gadget $W$ to encode a 3-SAT formula
$\cal F$.  We add two vertices $c$ and $d$ and the arcs $bc$, $cd$ and
$da$ to $W$ and call $W'$ the resulting digraph. We consider $W'$ as
an instance of {\sc $(\delta^+\geq 1,\delta^-\geq
  1)$-bipartite-partition-with-a-fixed-vertex} where we ask that $c\in
V_2$. Putting $c$ in $V_2$ forces $b$ and $d$ to be in $V_1$ and $a$
to be in $V_2$. So using Claim~\ref{claim4} we have that $\cal F$ is
satisfiable if and only if $D'$ is a positive instance of the problem.
\end{proof}




Finally if we just want a 2-partition $(V_1,V_2)$ so that every vertex in
$V_1$ has an out-neighbour in $V_2$ and every vertex in $V_2$ has a
neighbour (can be out- or in-) in $V_1$, then it turns out that such a
partition always exists.

\begin{theorem}
Every digraph $D$ with $\delta(D)\geq 1$ has a $(\delta^+\geq
1,\delta\geq 1)$-bipartite-partition.
\end{theorem}

\begin{proof}
Clearly we may assume that $D$ is connected.  Let $X_1$ contain one
vertex from each terminal component of $D$ (if $D$ is strong then $D$
itself is a terminal component and $|X_1|=1$). Note that $X_1$ is a
stable set.  Let $X_2$ be all vertices not in $X_1$ with an arc
into $X_1$.  Let $X_3$ be all vertices not in $X_1 \cup X_2$ with an
arc into $X_2$.  Let $X_4$ be all vertices not in $X_1 \cup X_2 \cup
X_3$ with an arc into $X_3$.  continue this process until some
$X_k=\emptyset$. Note that $V(D)=X_1 \cup X_2 \cup \cdots \cup
X_{k-1}$ as every vertex in $D$ has a path into a vertex from $X_1$.
Let $V_1$ contain all $X_i$ when $i$ is even and let $V_2$ contain all
$X_i$ when $i$ is odd.  Note that as every vertex in $X_{i}$ has an
arc into $X_{i-1}$ when $i>1$ and every vertex in $X_1$ has an arc
into it from $X_2$ as $D$ is connected. So the partition $(V_1,V_2)$
is a $(\delta^+\geq 1,\delta\geq 1)$-bipartite-partition of $D$.
\end{proof}

It could be interesting to extend the previous result to other
values of $k_1$ and $k_2$ or at least to determine the complexity of
finding such a partition.

\begin{question}
For any fixed pair $(k_1, k_2)$ of positive integers, what is the
complexity of deciding whether a given digraph has a $(\delta^+\geq
k_1,\delta\geq k_2)$-bipartite-partition?
\end{question}


\begin{thebibliography}{10}

\bibitem{alonCPC15}
N.~Alon.
\newblock {Splitting digraphs}.
\newblock {\em Combin. Probab. Comput.}, 15:933--937, 2006.

\bibitem{bangTCS526}
J.~Bang-Jensen and S.~Bessy.
\newblock {(Arc-)disjoint flows in networks}.
\newblock {\em Theoretical Computer Science}, 526:28--40, 2014.

\bibitem{bangman17}
J.~Bang-Jensen, S.~Bessy, F.~Havet, and A.~Yeo.
\newblock Out-degree reducing 2-partitions of digraphs, 2017.

\bibitem{bangTCS640}
J.~Bang-Jensen, N.~Cohen, and F.~Havet.
\newblock {Finding good 2-partitions of digraphs II. Enumerable properties}.
\newblock {\em Theoretical Computer Science}, 640:1--19, 2016.

\bibitem{bang2009}
J.~Bang-Jensen and G.~Gutin.
\newblock {\em {Digraphs: Theory, Algorithms and Applications, 2nd Edition}}.
\newblock Springer-Verlag, London, 2009.

\bibitem{bangTCS636}
J.~Bang-Jensen and F.~Havet.
\newblock {Finding good 2-partitions of digraphs I. Hereditary properties}.
\newblock {\em Theoretical Computer Science}, 636:85--94, 2016.

\bibitem{goncalvesDAM160}
D.~Gon\c{c}alves, F.~Havet, A.~Pinlou, and S.~Thomass\'e.
\newblock {On spanning galaxies in digraphs}.
\newblock {\em Discrete Applied Math.}, 160:744--754, 2012.

\bibitem{heggernesNJC5}
P.~Heggernes and J.A. Telle.
\newblock Partitioning graphs into generalized dominating sets.
\newblock {\em Nordic J. Comput.}, 5:128--142, 1998.

\bibitem{lovaszUM1973}
L.~Lov\`asz.
\newblock Coverings and colorings of hypergraphs.
\newblock In {\em Proc. 4th S.E. Conf. on Combinatorics, Graph Theory and
  Computing}, page 3–12. Utilitas Math., 1973.

\bibitem{robertsonAM150}
N.~Robertson, P.D. Seymour, and R.~Thomas.
\newblock {Permanents, Pfaffian orientations, and even directed circuits}.
\newblock {\em Ann. Math.}, 150:929--975, 1999.

\bibitem{shaeferSTOC10}
T.J. Schaefer.
\newblock {The complexity of satisfiability problems}.
\newblock In {\em Proceedings of the 10th Annual ACM Symposium on Theory of
  Computing (STOC 10)}, pages 216--226, New York, 1978. ACM.

\bibitem{thomassenEJC6}
C.~Thomassen.
\newblock {Even cycles in directed graphs}.
\newblock {\em Eur. J. Combin.}, 6(1):85--89, 1985.

\end{thebibliography}
\end{document}